\documentclass[11pt]{article}
\usepackage{stmaryrd}
\usepackage{amsfonts}
\usepackage{amssymb}
\usepackage{mathrsfs}
\usepackage{amsthm}
\usepackage{epsfig}\usepackage{amsmath}

\def\Reals{\mathop{\hbox{\mit I\kern-.2em R}}\nolimits}

\newtheorem{thm}{Theorem}
\newtheorem{lem}{Lemma}
\newtheorem{remark}{Remark}

\begin{document}

\title{A New Reduction from Search SVP to Optimization SVP}
\author{Gengran Hu, Yanbin Pan\\
Key Laboratory of Mathematics Mechanization\\
Academy of Mathematics and Systems Science, Chinese Academy of
Sciences\\Beijing 100190, China\\
 hudiran10@mails.gucas.ac.cn, panyanbin@amss.ac.cn\\}
\date{}
\maketitle

\begin{abstract}
It is well known that search SVP is equivalent to optimization SVP. However, the former
reduction from search SVP to optimization SVP by Kannan needs polynomial times calls to the
oracle that solves the optimization SVP. In this paper, a new rank-preserving reduction is presented with only one call to the optimization
SVP oracle.  It is obvious that the new reduction needs the least
calls, and improves Kannan's classical result. What's more, the idea also leads
a similar direct reduction  from search CVP to optimization CVP with only one call to the oracle.
\end{abstract}


\section{Introduction}
Given a matrix $B=(b_{ij})\in \mathbb{R}^{m\times
    n}$ with rank $n$, the lattice $L(B)$ spanned by
    the columns of $B$ is
    $$L(B)=\{\sum_{i=1}^n x_ib_i|x_i\in\mathbb{Z}\},$$
    where $b_i$ is the $i$-th column of $B$.  Lattice has many important
applications in cryptography. The shortest vector problem (SVP) and the closest vector
problem (CVP) are two of the most famous problems of lattice.\par

    SVP refers to find the shortest non-zero vector in a given lattice. There are three different variants of SVP:

    1. Search SVP: Given a lattice basis $B \in \mathbb{Z}^{m\times n}$, find $v \in
\mathcal
    {L}(B)$ such that $\|v\| = \lambda_1(\mathcal
    {L}(B))$, where $\lambda_1(\mathcal
    {L}(B))$ is the length of the shortest non-zero vector in $\mathcal
    {L}(B)$.

    2. Optimization SVP: Given a lattice basis $B \in \mathbb{Z}^{m\times n}$, find
$\lambda_1(\mathcal
    {L}(B))$.

    3. Decisional SVP: Given a lattice basis $B \in \mathbb{Z}^{m\times n}$ and a rational
$r\in \mathbb{Q}$, decide whether $\lambda_1(\mathcal
    {L}(B))\leq r$ or not.

It has been proved that the three problems are equivalent to each other (see \cite{mic}). It is easy to check that
the decisional SVP is as hard as the optimization SVP and the optimization variant can be reduced to the search
variant.\par

In 1987, Kannan \cite{kan} also showed that the  search variant can be
reduced to the optimization variant.  The basic idea of his reduction is to recover the integer
coefficients of some shortest vector under the given lattice basis by
introducing small errors to the original lattice basis. However, his reduction is a bit complex. It
needs to call polynomial times optimization SVP oracle, since it could not determine the signs
of the shortest vector's entries  at one time. It also
needs oracle to solve optimization SVP for some  lattices with lower rank besides with the same rank as the original lattice. \par

In this paper, we propose a new rank-preserving reduction which can solve the search SVP with only  one call to the optimization
SVP oracle. It is obvious that there is no reduction with less calls than ours. Instead of recovering the shortest vector directly as in \cite{kan}, we
first recover the integer coefficients of some shortest vector under the given lattice basis, then recover the shortest vector.\par
 A similar direct reduction from  search CVP to optimization CVP with only one call also holds whereas some popular reductions \cite{mic, reg} usually takes decisional CVP to bridge the
 search CVP and optimization CVP.  The former reduction from decisional CVP to optimization CVP needs one call to the optimization CVP  oracle, but it
 needs polynomial times calls to the decisional CVP oracle to reduce  search CVP to decisional CVP.


\section{The New Reduction}
For simplicity, we just give the new reduction for the full rank lattice, i.e. $n=m$, as in \cite{kan}. It is easy to general
the new reduction for the lattices with rank $n<m$.

\subsection{Some Notations}
    Given a lattice basis $B=(b_{ij})\in \mathbb{R}^{n\times
    n}$, let $M(B)=\max|b_{ij}|$.   For lattice $L(B)$, we define its SVP solution set $S_{B}$ as:
$$
        S_{B}=\{x\in \mathbb{Z}^n |\|Bx\|=\lambda_1(\mathcal{L}(B))\}
$$

    Denote by $poly(n)$ the polynomial in  $n$.

\subsection{Some Lemmas}
We need some lemmas to prove our main theorem.
   \begin{lem}
    \label{lem1} For every positive integer $n$, there exist $n$ positive integers $a_{1}<a_{2}<\ldots <a_{n}$
    s.t. all the $a_{i}+a_{j}(i\leq j)$'s are distinct
    and $a_{n}$ is bounded by $poly(n)$.
    \end{lem}

    \begin{proof} We can take $a_{k}=(n^2+k-1)^2$ for $k=1, 2, \cdots,n$. Suppose
    $a_{i_{1}}+a_{j_{1}}=a_{i_{2}}+a_{j_{2}}$ for some $i_{1}, j_{1}, i_{2}, j_{2}$,
    we get
    $(i_{1}-1)^2+(j_{1}-1)^2+2n^2((i_{1}-1)+(j_{1}-1))=(i_{2}-1)^2+(j_{2}-1)^2+2n^2((i_{2}-1)+(j_{2}-1))$.
  Since $(i_{1}-1)^2+(j_{1}-1)^2,(i_{2}-1)^2+(j_{2}-1)^2<2n^2$,  we have $(i_{1}-1)^2+(j_{1}-1)^2=(i_{2}-1)^2+(j_{2}-1)^2$
    and $i_{1}+j_{1}=i_{2}+j_{2}$, which leads $\{i_{1}, j_{1}\}=\{i_{2}, j_{2}\}$. Hence all the $a_{i}+a_{j}(i\leq
    j)$'s are distinct. It is obvious that $a_{n}\leq(n^2+n-1)^2$.
    \end{proof}

    \begin{lem}
    \label{lem2} Given positive odd integer $p>2$, and any positive integer $n$, which satisfies $n=\sum_{i=0}^{k}n_ip^i$ where
    $|n_i|\leq\lfloor p/2\rfloor$, then we can recover the coefficients $n_i$'s in polynomial time.
    \end{lem}

    \begin{proof} We can recover $n_0$ by computing $a \equiv n\mbox{ mod } p$ and choose $a$ in the interval from $-\lfloor p/2\rfloor$
    to $\lfloor p/2\rfloor$. After obtaining $n_0$, we get another integer $(n-n_0*p^0)/p$. Recursively, we can recover all the coefficients. This can be done in polynomial time obviously.
    \end{proof}

   \begin{lem}
    \label{lem3}  For bivariate polynomial $f(x,y)=xy$,
    given any lattice basis matrix $B\in \mathbb{Z}^{n\times n}$, $\lambda_1(L(B))$ has an upper
    bound $f(M,n)$, where $M=M(B)$. What's more,  for every $x\in S_{B}$, $|x_i|$ $(i=1,\cdots, n)$ has an upper
    bound $f(M^n,n^n)$.
    \end{lem}

    \begin{proof} The length of any column of $B$ is an upper bound of $\lambda_1(L(B))$, so $\lambda_1(L(B))\leq n^{1/2}M\leq nM$.\par
    For $x\in S_B$, we let $y=Bx$, then $\|y\|=\lambda_1(L(B))\leq \sqrt{n}M$. By Cramer's rule, we know that
    $$x_i=\dfrac{\det(B^{(i)})}{\det(B)},$$
    where $B^{(i)}$ is formed by replacing the $i$-th column of $B$ by $y$. By Hadamard's inequality, $|\det(B^{(i)})|\leq n^{n/2}M^n\leq n^nM^n$. We know $|\det(B)|\geq 1$ since $\det(B)$ is a non-zero integer. Hence $|x_i|\leq n^nM^n$.
    \end{proof}

\subsection{The Main Theorem}
\begin{thm}  Assume there exists an  oracle $\mathcal{O}$ that can solve the optimization SVP
for any lattice $L(B^\prime)$ with basis $B^\prime\in\mathbb{Z}^{n\times n}$, then there is an algorithm that
can solve the search SVP for any lattice $L(B)$ with basis $B\in\mathbb{Z}^{n\times n}$ with only one call to $\mathcal{O}$ in $poly(\log_2{M},n,\log_2{n})$  time, where $M=M(B)$.
\end{thm}

\begin{proof}
The main steps of the algorithm are as below:\\
(1) Constructing a new lattice basis $B_{\epsilon}\in\mathbb{Z}^{n\times n}$.

    We construct $B_{\epsilon}$ from the original
    lattice $B$:
   $$
    B_{\epsilon}=\epsilon_{n+1}B+ \left(                 
    \begin{array}{cccc}   
    \epsilon_{1} & \epsilon_{2} & \dots & \epsilon_{n}\\  
    0 & 0 & \dots & 0\\  
    \vdots & \vdots &  & \vdots\\
    0 & 0 & \dots & 0\\
    \end{array}
    \right)                
    $$
    where the $\epsilon_{i}$ will be determined as below. \par
     For any $x\in \mathbb{Z}^n$, we difine $c(x)=\sum_{i=1}^{n}b_{1i}x_{i}$. For $x\in S_B$,  by Lemma \ref{lem3}, $|x_i|$ has an upper
    bound $f(M^n,n^n)$. Let $M_{1}=2f((M+1)^n,n^n)$. In addition, $\|Bx\|=\lambda_1(L(B))$  is bounded by
    $f(M,n)$. Let $M_{2}=f(M+1,n)$. $|c(x)|$  is also bounded by $M_{2}$ since $|c(x)|\leq\|Bx\|$.  We let
     $$p=2*\max{\{M_2^2,2M_1M_2,2M_1^2\}}+1.$$

    By Lemma \ref{lem1}, we can choose $n+1$ positive integers $a_{1}<a_{2}<\ldots <a_{n+1}$,
    such that all the $a_{i}+a_{j}(i\leq j)$'s are distinct
    where $a_{n+1}$ is bounded by $poly(n)$. Let
  $$
    \epsilon_{i}=p^{a_{i}}.
  $$
  \par
   We first show that $|\det{(\frac{1}{\epsilon_{n+1}}B_\epsilon)|\geq \frac{1}{2}}$, so $B_\epsilon$ is indeed a lattice basis. Notice that
   $$\det{(\frac{1}{\epsilon_{n+1}}B_\epsilon)}=\det(B)+\sum_{i=1}^n\alpha_i\frac{\epsilon_i}{\epsilon_{n+1}},$$
   where $\alpha_i$ is the cofactor of $B_{1i}$ in $B$. Since $\frac{\epsilon_i}{\epsilon_{n+1}}\leq \frac{1}{p^2}$ and $|\alpha_i|\leq M^{n-1}(n-1)^{n-1}$, $|\sum_{i=1}^n\alpha_i\frac{\epsilon_i}{\epsilon_{n+1}}|\leq \frac{1}{p^2}M^{n-1}n^n<\frac{1}{2}$. By the fact $\det(B)$ is a non-zero integer, we get
  \begin{equation}\label{equ1}
   |\det{(\frac{1}{\epsilon_{n+1}}B_\epsilon)|\geq \frac{1}{2}}.
   \end{equation}
   
   We claim that $S_{B_\epsilon}\subseteq S_B$. Since $S_{B_\epsilon}= S_{\frac{1}{\epsilon_{n+1}}B_\epsilon}$, it is enough to prove $S_{\frac{1}{\epsilon_{n+1}}B_\epsilon}\subseteq S_B$. \par

   For any $x\in S_{\frac{1}{\epsilon_{n+1}}B_\epsilon}$, by (\ref{equ1}) and the proof of Lemma \ref{lem3}, we know that
   $|x_i|\leq M_1$, $|c(x)|\leq M_2$.  By the choice of $p$,  $ x_{i}^2, 2c(x)x_{i}, 2x_{i}x_{j}$ are in the interval $[-\lfloor p/2 \rfloor, \lfloor p/2 \rfloor]$. Together with the fact that $\frac{\epsilon_{i}\epsilon_{j}}{\epsilon_{n+1}^2}(i\leq j)$'s are different powers of $p$,  we have
   \begin{equation}\label{equ2}
    \begin{array}{rcl}
   \lambda_1(L(\frac{1}{\epsilon_{n+1}}B_\epsilon))^2 &=& \|\frac{1}{\epsilon_{n+1}}B_\epsilon x\|^2\\
  &=&\|Bx\|^2+\sum_{i=1}^{n}x_{i}^2(\frac{\epsilon_{i}}{\epsilon_{n+1}})^2+\sum_{i=1}^{n}2c(x)x_{i}\frac{\epsilon_{i}}{\epsilon_{n+1}}
   +\sum_{i<j}2x_{i}x_{j}\frac{\epsilon_{i}\epsilon_{j}}{\epsilon_{n+1}^2} \\
         &> &  \|Bx\|^2 -(\lfloor p/2 \rfloor+1)\frac{\epsilon_{n}}{\epsilon_{n+1}}.
    \end{array}
   \end{equation}
   Similarly, for any $y\in S_B$, we have
   \begin{equation}\label{equ3}
     \begin{array}{rcl}
     \|\frac{1}{\epsilon_{n+1}}B_\epsilon y\|^2  &=& \|By\|^2+\sum_{i=1}^{n}y_{i}^2(\frac{\epsilon_{i}}{\epsilon_{n+1}})^2+\sum_{i=1}^{n}2c(y)y_{i}\frac{\epsilon_{i}}{\epsilon_{n+1}}
   +\sum_{i<j}2y_{i}y_{j}\frac{\epsilon_{i}\epsilon_{j}}{\epsilon_{n+1}^2}  \\
       &<& \lambda_1(L(B))^2 +(\lfloor p/2 \rfloor+1)\frac{\epsilon_{n}}{\epsilon_{n+1}}
       \end{array}
   \end{equation}
 Next, we prove $S_{\frac{1}{\epsilon_{n+1}}B_\epsilon}\subseteq S_B$. Suppose there exists $x\in S_{\frac{1}{\epsilon_{n+1}}B_\epsilon}$ but $x\not\in S_B$, then  \begin{equation}\label{equ4}
 \|Bx\|^2\geq \lambda_1(L(B))^2+1.
   \end{equation}
  Notice that $\frac{\epsilon_{n}}{\epsilon_{n+1}}<\frac{1}{p^2}$, we have $0<(\lfloor p/2 \rfloor+1)\frac{\epsilon_{n}}{\epsilon_{n+1}}<\frac{1}{2}$. Together with (\ref{equ2}), (\ref{equ3}) and (\ref{equ4}), we have
  $$
  \begin{array}{rcl}
        \lambda_1(L(\frac{1}{\epsilon_{n+1}}B_\epsilon))^2 &> &  \|Bx\|^2 -(\lfloor p/2 \rfloor+1)\frac{\epsilon_{n}}{\epsilon_{n+1}}\\
        &\geq& \lambda_1(L(B))^2 +1-(\lfloor p/2 \rfloor+1)\frac{\epsilon_{n}}{\epsilon_{n+1}} \\
        &>& \lambda_1(L(B))^2 +(\lfloor p/2 \rfloor+1)\frac{\epsilon_{n}}{\epsilon_{n+1}}\\
        &>& \|\frac{1}{\epsilon_{n+1}}B_\epsilon y\|^2,
   \end{array}
  $$
 which is an contradiction, since $\frac{1}{\epsilon_{n+1}}B_\epsilon y\in L(\frac{1}{\epsilon_{n+1}}B_\epsilon)$. Hence $S_{B_\epsilon}\subseteq S_B$.\\
(2) Querying the oracle $\mathcal{O}$ with $B_\epsilon$ once, we get
    $\lambda_1(\mathcal{L}(B_{\epsilon}))$. \par
   So there exists $x=(x_{1},\ldots, x_{n})^{\mathrm{T}}\in S_{B_{\epsilon}}\subseteq S_B$, such that
 $$
   \|Bx\|^2\epsilon_{n+1}^2+\sum_{i=1}^{n}x_{i}^2\epsilon_{i}^2+\sum_{i=1}^{n}2c(x)x_{i}\epsilon_{n+1}\epsilon_{i}
   +\sum_{i<j}2x_{i}x_{j}\epsilon_{i}\epsilon_{j}=\lambda_1(\mathcal{L}(B_{\epsilon}))^2
$$
(3) Recovering all the $x_i$'s and output $Bx$.\par
Since $x\in S_B$, every coefficient  $\|Bx\|^2, x_{i}^2, 2c(x)x_{i}, 2x_{i}x_{j}$ is in the interval $[-\lfloor p/2 \rfloor, \lfloor p/2 \rfloor]$ and
$\epsilon_i\epsilon_j$ $(i\leq j)$'s are different powers of $p$.
Hence, $\log_2{(\lambda_1(\mathcal{L}(B_{\epsilon})))}$ is bounded by $poly(\log_2M,n,\log_2n)$. Furthermore,
    by Lemma \ref{lem2}, we can recover all the coefficients in  $poly(\log_2M,n,\log_2n)$ time.
    Especially, we can recover  all $x_{i}^2$ and $x_{i}x_{j}(i\neq j)$. Let $k=\min\{i|x_i\neq0\}$. We fix $x_k=\sqrt{x_k^2}>0$, and can recover all
    the remaining  $x_j=sign(x_{k}x_{j})\sqrt{x_{j}^2}$ according to  $x_{j}^2$ and $x_{k}x_{j}(k\neq j)$.\par

It is easy to check that  the complexity of every step is bounded by $poly(\log_2{M},n,\log_2n)$.
\end{proof}

\begin{remark}
For any search CVP instant $(B, t)$, given an oracle which can solve the optimization CVP,  we can call the oracle with
$(B_\epsilon, \epsilon_{n+1}t)$ only once to solve the search CVP similarly.
\end{remark}

\section{Conclusions}

In this paper, we give a new reduction from search SVP to optimization SVP with only one call,  which is the least, to the  optimization SVP  oracle.
A similar result for CVP also holds. However,  it seems hard to apply the idea for GapSVP or GapCVP, since the new reduction is also
sensitive to the error.

\end{document}